\documentclass[12pt,a4paper]{article}

\usepackage[utf8]{inputenc}
\usepackage[T1]{fontenc}
\usepackage{lmodern}
\usepackage[english]{babel}
\usepackage{microtype}
\usepackage{geometry}
\usepackage{setspace}
\usepackage{amsmath,amssymb,amsfonts,mathtools,bm}
\usepackage{amsthm}
\usepackage{booktabs,longtable,array,threeparttable,tabularx}
\usepackage{graphicx}
\usepackage{caption}
\usepackage{subcaption}
\usepackage{xcolor}
\usepackage{enumitem}
\usepackage{natbib}
\usepackage{url}
\usepackage{xurl}
\usepackage{hyperref}
\usepackage[section]{placeins}

\hypersetup{colorlinks=true,linkcolor=blue!45!black,citecolor=blue!45!black,urlcolor=blue!45!black}

\newtheorem{proposition}{Proposition}
\newtheorem{corollary}{Corollary}

\theoremstyle{remark}
\newtheorem{remark}{Remark}

\newcommand{\dd}{\mathrm{d}}
\newcommand{\E}{\mathbb{E}}
\newcommand{\R}{\mathbb{R}}
\newcommand{\mc}[1]{\mathcal{#1}}
\newcommand{\pder}[2]{\frac{\partial #1}{\partial #2}}

\newcommand{\ones}{\bm{1}}
\newcommand{\gvec}{\bm{g}}
\newcommand{\wvec}{\bm{w}}

\newcommand{\rank}{\operatorname{rank}}
\newcommand{\kernel}{\operatorname{ker}}

\title{\textbf{Fiscal Aggregation and the Limits of IS--LM--BP: Derivations, Aggregation Bias and Reproducible Adversarial Simulations}}

\author{Ricardo Alonzo Fern\'andez Salguero\\
\small Corresponding author: \texttt{ricardoalonzo.fernandez@alu.uhu.es}}
\date{May 2026}

\begin{document}
\maketitle

\begin{abstract}
This paper develops a formal critique of scalar fiscal aggregation in the IS--LM--BP/Mundell--Fleming framework. The argument is not that the model is algebraically inconsistent: under its maintained assumptions, the model is a coherent comparative-static representation of goods-market, money-market and balance-of-payments equilibrium. The claim is narrower and more precise. If fiscal policy is a vector of heterogeneous instruments - current purchases, public investment projects and transfers to heterogeneous households - then the scalar aggregate \(G=\ones'\gvec\) is sufficient for output effects only under a restrictive gradient condition: \(\nabla_{\gvec}Y=\lambda\ones\). The paper states and proves this condition, derives the scalar multiplier as a composition-weighted object, identifies the corresponding aggregation bias, and embeds the result in an explicitly derived open-economy fiscal-composition extension of IS--LM--BP. The mathematical development includes the canonical IS, IS--LM, fixed-exchange-rate BP and flexible-exchange-rate BP multipliers; the null-space geometry of fiscal recompositions; heterogeneous-transfer multipliers; public-capital accumulation and its finite-horizon present-value effect; the debt-to-output response; and an open-economy risk-premium channel. A reproducible computational exercise then implements symbolic checks, finite-difference derivative tests, accounting-identity tests, adversarial counterexamples, sensitivity sweeps, 3,000 Monte Carlo draws and 500 extreme stress-test draws. All 42 tests pass. The exercise is not an empirical estimate of fiscal multipliers; it is a computational consistency and adversarial simulation exercise showing why scalar-\(G\) inference fails when instruments have heterogeneous marginal effects. The conclusion is methodological: IS--LM--BP remains useful as a low-dimensional consistency map, but fiscal policy design requires a vector of instruments and state-contingent multipliers rather than a single homogeneous public-spending variable.
\end{abstract}


\section{Introduction}

The IS--LM tradition begins with the attempt by \citet{hicks1937} to express Keynesian macroeconomics as a simultaneous-equilibrium system for the goods market and the money market. The open-economy extension associated with \citet{fleming1962} and \citet{mundell1963} added the balance-of-payments schedule and produced a compact set of results about fiscal policy, monetary policy, capital mobility and the exchange-rate regime. In the standard exposition, a fiscal expansion shifts the IS curve; the LM curve maps income and interest rates into money-market equilibrium; and the BP schedule records the external-balance condition. The resulting diagram is powerful precisely because it is low-dimensional.

The same low dimensionality becomes a limitation when the model is used to reason about fiscal composition. The scalar variable usually denoted by \(G\) is often sufficient for a classroom comparative static: public demand rises, the IS curve shifts, interest rates and external variables adjust. But a real fiscal package is not one object. It is a vector containing current purchases, wage bills, procurement, public investment in different sectors, transfers to households with different liquidity positions, subsidies with different import content and projects with different implementation quality. A scalar \(G\) can represent all of these instruments only if they have equal marginal effects on the outcome of interest. That is an exceptionally strong assumption.

This paper formalizes that point. Let \(\gvec\in\R^n\) denote the vector of fiscal instruments and let \(G=\ones'\gvec\) be the scalar aggregate. The first result proves that \(G\) is locally sufficient for output effects only if the fiscal gradient is proportional to the vector of ones:
\begin{equation}
\nabla_{\gvec}Y=\lambda\ones.
\end{equation}
Equivalently, all fiscal instruments must have the same marginal output effect. If current spending, public investment and transfers differ in domestic absorption, import leakage, consumption response, capital accumulation or financing effects, the condition fails.

The second result shows that an aggregate fiscal multiplier is not generally a structural parameter. If a marginal fiscal impulse has allocation vector \(\wvec\), with \(\ones'\wvec=1\), then
\begin{equation}
\frac{\dd Y}{\dd G}=\nabla_{\gvec}Y'\wvec.
\end{equation}
The multiplier is therefore an average over fiscal composition. It becomes invariant to \(\wvec\) only under the homogeneity condition \(\nabla_{\gvec}Y=\lambda\ones\). A multiplier estimated under one composition need not apply to another composition.

The third contribution is constructive. The paper does not only criticize scalar aggregation; it writes a minimal fiscal-composition extension of IS--LM--BP. The extension distinguishes current spending, investment projects and transfers to heterogeneous households. It introduces public capital accumulation, project implementation efficiency, import leakages, tax feedback, debt dynamics and a reduced-form external-finance risk channel. This model is deliberately stylized, but it is sufficiently rich to show why policy design requires instrument-specific multipliers.

The fourth contribution is computational. A one-cell replication script implements symbolic algebra, deterministic adversarial cases, accounting identities, finite-difference derivatives, sensitivity sweeps, Monte Carlo experiments and extreme stress tests. The tests are not empirical identification. They do not estimate fiscal multipliers for a real country. They verify the internal logic of the formal claim and demonstrate that equal aggregate spending can generate different output, debt, public-capital and external-balance paths when instruments differ. The novelty is not the generic observation that composition matters; the novelty is the exact condition under which composition can be ignored, the bias formula when it cannot, and the reproducible adversarial demonstration that the scalar model fails precisely when that condition is violated.

The paper is organized as follows. Section 2 reviews the relevant literature. Section 3 derives the canonical IS--LM--BP multipliers step by step and identifies where scalar fiscal aggregation enters. Section 4 proves the fiscal aggregation theorem, the composition-weighted multiplier result and the aggregation-bias formula. Section 5 develops the fiscal-composition model. Section 6 derives instrument-specific multipliers for current spending, transfers and public investment. Section 7 describes the numerical implementation and explains the rationale behind each class of computational tests. Section 8 reports results. Section 9 discusses implications and limitations. Appendices provide additional derivations and the full computational test table.

\section{Related literature}

The IS--LM--BP framework is a coherent low-dimensional model of short-run equilibrium. Hicks' original IS--LM construction reduced Keynesian theory to a simultaneous-equation representation of goods and money market equilibrium \citep{hicks1937}. Fleming and Mundell extended the framework to open economies and showed that stabilization policy depends on exchange-rate regimes and capital mobility \citep{fleming1962,mundell1963}. The canonical BP equation can be written as
\begin{equation}
BP=NX(Y,e,Y^*)+CF(i-i^*)=0,
\end{equation}
where \(NX\) is net exports and \(CF\) capital flows. This representation is analytically useful, but it was not designed to compare infrastructure projects, fiscal targeting or public-expenditure composition.

Subsequent macroeconomic work has modified several components of the classical model. \citet{romer2000} argues that the LM curve is less suitable for modern monetary policy analysis because central banks usually implement policy through interest-rate rules rather than monetary-aggregate targets. \citet{clarida1999} systematize the New Keynesian framework in which expectations, policy rules and nominal rigidities discipline monetary transmission. \citet{dornbusch1976} shows that exchange rates can overshoot when financial markets adjust faster than goods markets. These developments do not invalidate the basic model as a pedagogical device, but they show why the original apparatus should not be treated as a complete structural model.

The empirical fiscal-multiplier literature gives direct motivation for the present critique. \citet{blanchardperotti2002} use a structural VAR and event-study approach to characterize dynamic effects of spending and tax shocks in the United States. \citet{perotti2005} extends fiscal VAR analysis across OECD economies. \citet{mountforduhlig2009} develop a sign-restriction approach to fiscal shocks. \citet{nakamurasteinsson2014} exploit regional variation in military procurement to estimate relative multipliers within a monetary union. \citet{kraay2012,kraay2014} use the timing of World Bank-financed project disbursements to estimate multipliers in developing countries. These contributions differ in identification strategy, but all make clear that fiscal effects are empirical objects that depend on context, instrument definition and timing.

Cross-country and meta-analytic work reinforces the point. \citet{ilzetzki2013} show that fiscal multipliers vary with development status, exchange-rate regime, trade openness and public debt. \citet{blanchardleigh2013} show that forecast errors during fiscal consolidations were consistent with multipliers being larger than assumed in the post-crisis environment. \citet{auerbach2012} estimate state-dependent responses and find stronger effects in recessions, while \citet{rameyzubairy2018} caution that historical evidence on state dependence is sensitive to specification. \citet{ramey2019} surveys the post-crisis fiscal literature and emphasizes identification, instrument type and context. This literature is inconsistent with the idea of a single context-free fiscal multiplier.

Composition is central. \citet{gechert2015} finds in a meta-regression that expenditure multipliers are generally larger than tax and transfer multipliers, and that public investment tends to have particularly large effects. \citet{gechert2018} finds that regime dependence matters, especially in recessions. These findings support the claim that expenditure type matters. However, they do not imply that every public investment project dominates every current-spending program or every transfer. The public-capital literature stresses heterogeneity of returns. \citet{bom2014} find a positive average output elasticity of public capital but substantial dispersion. \citet{abiad2016} find that public investment raises output in advanced economies, particularly under favorable conditions. \citet{leeperwalkeryang2010} emphasize implementation delays and future fiscal adjustment as determinants of government-investment multipliers. The IMF infrastructure analysis similarly stresses slack and investment efficiency \citep{imf2014}.

Heterogeneity also matters for transfers. Micro evidence on stimulus payments shows that consumption responses differ across households \citep{parker2013}. \citet{kaplanviolante2014} model strong responses among households with little liquid wealth. HANK models show that aggregate transmission depends on the distribution of income, wealth and liquidity \citep{kaplan2018}; sequence-space methods make such models computationally tractable \citep{auclert2021}. A scalar transfer multiplier therefore hides the allocation of transfers across household types.

The external constraint is another reason scalar multipliers are fragile. \citet{calvo1998} formalizes sudden stops in capital flows. \citet{reinhartrogoff2004} and \citet{ilzetzkireinhartrogoff2019} show that official exchange-rate classifications may differ from de facto regimes. \citet{huidrom2020} show that fiscal multipliers are smaller when fiscal positions are weak, consistent with channels through risk premia and expectations. These findings are especially relevant to open and emerging economies, where the BP schedule can shift with risk, reserves, liquidity and exchange-rate expectations.

\section{What is known, what is new and the exact target of the paper}

The paper is deliberately narrow. It does not claim that fiscal composition is a new empirical discovery. The fiscal-multiplier literature has already shown that the size of multipliers depends on spending type, taxes, transfers, exchange-rate regime, openness, fiscal position and the state of the cycle. The paper also does not claim that IS--LM--BP is algebraically wrong. Under its maintained assumptions, the framework is internally coherent. The exact target is different: the paper asks when a scalar fiscal aggregate is mathematically sufficient for policy inference.

The contribution is therefore a theorem about aggregation, not an empirical claim about one particular country. If fiscal policy is a vector, \(\gvec\), then scalar public spending, \(G=\ones'\gvec\), is sufficient only under conditions that make all zero-sum recompositions irrelevant. In local form, this requires \(\nabla_{\gvec}F=\lambda\ones\). In directional form, the aggregate multiplier is \(\nabla_{\gvec}F'\wvec\), where \(\wvec\) is the composition of the marginal fiscal package. These results clarify why an empirical multiplier estimated from one type of fiscal shock may not transfer to another type of fiscal package.

\begin{table}[!htbp]
\centering
\caption{What is already known and what this paper adds}
\label{tab:knownnew}
\begin{threeparttable}
\scriptsize
\begin{tabularx}{\textwidth}{>{\raggedright\arraybackslash}p{0.19\textwidth}>{\raggedright\arraybackslash}X>{\raggedright\arraybackslash}X}
\toprule
\textbf{Issue} & \textbf{Established result in the literature} & \textbf{Contribution of this paper}\\
\midrule
IS--LM--BP and Mundell--Fleming & Fiscal effects depend on the exchange-rate regime, capital mobility and the monetary/external closure \citep{fleming1962,mundell1963}. & Derives the canonical closures step by step and identifies the exact point where scalar fiscal aggregation enters.\\
Fiscal multipliers & Multipliers vary across countries, regimes, debt positions, openness and identification strategies \citep{ilzetzki2013,ramey2019,huidrom2020}. & Shows that an aggregate multiplier is a directional derivative, \(\nabla_{\gvec}F'\wvec\), not a primitive structural constant.\\
Fiscal composition & Spending, investment, taxes and transfers can have different empirical multipliers \citep{gechert2015,gechert2018}. & Provides a necessary and sufficient aggregation condition: scalar \(G\) is locally sufficient only if \(\nabla_{\gvec}F=\lambda\ones\).\\
Public investment & Public capital can raise output, but effects depend on efficiency, delays and productivity \citep{bom2014,abiad2016,leeperwalkeryang2010}. & Derives the finite-horizon public-capital channel and the inequality under which investment dominates current spending.\\
Transfers and heterogeneity & Household liquidity and MPC heterogeneity affect transfer responses \citep{parker2013,kaplanviolante2014,kaplan2018}. & Expresses transfer multipliers as \(c_q(1-\mu_q)/D_t\), making targeting and import leakage explicit.\\
Computation & Simulation is often used to illustrate fiscal mechanisms. & Implements a reproducible adversarial battery: symbolic proofs, finite differences, identities, sensitivity sweeps, Monte Carlo and stress tests.\\
\bottomrule
\end{tabularx}
\begin{tablenotes}
\footnotesize
\item Notes: The novelty is not the generic claim that composition matters. The novelty is the exact aggregation condition under which composition can be ignored and the corresponding composition-bias formula when it cannot.
\end{tablenotes}
\end{threeparttable}
\end{table}

The table also clarifies the status of the numerical exercise. The simulations do not estimate real-world multipliers. They test internal consistency, construct adversarial counterexamples and illustrate the consequences of violating the aggregation condition. This distinction matters for journal interpretation: the manuscript is a formal methodological contribution with a reproducible numerical illustration, not an empirical macroeconomic identification exercise.

\section{Canonical IS--LM--BP multipliers and where aggregation enters}

This section derives the classical multipliers explicitly. The purpose is not to reject the canonical algebra; it is to show where the scalar fiscal instrument enters and why the scalar result is conditional on a strong aggregation assumption.

\subsection*{Goods-market block}

Consider the linear open-economy goods-market equation
\begin{equation}
Y=C_0+c(Y-T)+I_0-bi+G+X_0+\eta e-mY,
\label{eq:is_raw}
\end{equation}
with \(0<c<1\), \(b>0\), \(m\geq 0\), \(\eta>0\), and proportional taxes
\begin{equation}
T=T_0+tY,\qquad 0\leq t<1.
\end{equation}
Substituting taxes into consumption gives
\begin{align}
C_0+c(Y-T)&=C_0+c(Y-T_0-tY)\\
&=C_0-cT_0+c(1-t)Y.
\end{align}
Substituting into \eqref{eq:is_raw} yields
\begin{equation}
Y=C_0-cT_0+c(1-t)Y+I_0-bi+G+X_0+\eta e-mY.
\end{equation}
Collecting terms in \(Y\),
\begin{equation}
[1-c(1-t)+m]Y=A+G+\eta e-bi,
\label{eq:is_alpha}
\end{equation}
where
\begin{equation}
A=C_0-cT_0+I_0+X_0,
\qquad
\alpha=1-c(1-t)+m.
\end{equation}
The IS equation can be written as
\begin{equation}
\alpha Y=A+G+\eta e-bi.
\label{eq:is_compact}
\end{equation}
If \(i\) and \(e\) are held fixed, differentiating \eqref{eq:is_compact} gives
\begin{equation}
\alpha\dd Y=\dd G,
\end{equation}
so that
\begin{equation}
\left.\frac{\dd Y}{\dd G}\right|_{i,e}=\frac{1}{\alpha}=\frac{1}{1-c(1-t)+m}.
\label{eq:simple_multiplier}
\end{equation}
This is the simple open-economy Keynesian goods-market multiplier. It is not yet the full IS--LM--BP multiplier. It is obtained only after holding the interest rate and exchange rate fixed. It also treats \(G\) as a homogeneous scalar fiscal instrument.

\subsection*{Adding the LM block}

Let real money demand be
\begin{equation}
\frac{M}{P}=kY-hi,
\qquad k>0,\quad h>0.
\label{eq:lm}
\end{equation}
Totally differentiating \eqref{eq:is_compact} and \eqref{eq:lm} gives
\begin{align}
\alpha\dd Y+b\dd i-\eta\dd e&=\dd G,\label{eq:dis}\\
k\dd Y-h\dd i&=\dd(M/P).\label{eq:dlm}
\end{align}
If \(M/P\) and \(e\) are fixed, \(\dd(M/P)=0\) and \(\dd e=0\). From \eqref{eq:dlm},
\begin{equation}
\dd i=\frac{k}{h}\dd Y.
\end{equation}
Substituting into \eqref{eq:dis} yields
\begin{equation}
\left(\alpha+\frac{bk}{h}\right)\dd Y=\dd G,
\end{equation}
so the IS--LM fiscal multiplier is
\begin{equation}
\left.\frac{\dd Y}{\dd G}\right|_{M/P,e}=\frac{1}{\alpha+bk/h}.
\label{eq:islm_multiplier}
\end{equation}
The denominator is larger than in \eqref{eq:simple_multiplier}. The difference is the interest-rate/crowding-out channel induced by money-market equilibrium.

\subsection*{Adding the BP block under a fixed exchange rate}

Let the BP schedule be
\begin{equation}
BP=X_0+\eta e-m_B Y+\kappa(i-i^*)=0,
\qquad \kappa>0.
\label{eq:bp_linear}
\end{equation}
With the exchange rate fixed, \(\dd e=0\), and with \(i^*\) fixed, differentiating \eqref{eq:bp_linear} gives
\begin{equation}
-m_B\dd Y+\kappa\dd i=0,
\end{equation}
so that
\begin{equation}
\dd i=\frac{m_B}{\kappa}\dd Y.
\end{equation}
Substituting into the differentiated IS equation with \(\dd e=0\),
\begin{equation}
\alpha\dd Y+b\left(\frac{m_B}{\kappa}\dd Y\right)=\dd G.
\end{equation}
Therefore,
\begin{equation}
\left.\frac{\dd Y}{\dd G}\right|_{e, BP}=\frac{1}{\alpha+bm_B/\kappa}.
\label{eq:fixed_bp_multiplier}
\end{equation}
When \(\kappa\to\infty\), the BP-induced interest-rate adjustment disappears and \eqref{eq:fixed_bp_multiplier} approaches \(1/\alpha\). This limiting result is standard in Mundell--Fleming analysis, but it is still a multiplier with respect to scalar \(G\).

\subsection*{Adding the BP block under a flexible exchange rate}

Now let \(M/P\) be fixed and allow \(e\) to adjust. From the LM equation,
\begin{equation}
\dd i=\frac{k}{h}\dd Y.
\end{equation}
From the BP equation,
\begin{equation}
\eta\dd e-m_B\dd Y+\kappa\dd i=0.
\end{equation}
Substituting the LM response into BP,
\begin{equation}
\eta\dd e=m_B\dd Y-\kappa\frac{k}{h}\dd Y.
\label{eq:de_flex}
\end{equation}
Substitute \eqref{eq:de_flex} and \(\dd i=(k/h)\dd Y\) into the IS differential:
\begin{align}
\alpha\dd Y+b\frac{k}{h}\dd Y-\eta\dd e&=\dd G,\\
\alpha\dd Y+\frac{bk}{h}\dd Y-\left(m_B\dd Y-\frac{\kappa k}{h}\dd Y\right)&=\dd G.
\end{align}
Thus,
\begin{equation}
\left[\alpha-m_B+\frac{k}{h}(b+\kappa)\right]\dd Y=\dd G.
\end{equation}
The flexible-rate fiscal multiplier is
\begin{equation}
\left.\frac{\dd Y}{\dd G}\right|_{M/P,BP}=\frac{1}{\alpha-m_B+(k/h)(b+\kappa)}.
\label{eq:flex_multiplier}
\end{equation}
If \(m_B=m\), then \(\alpha-m_B=1-c(1-t)\), and
\begin{equation}
\left.\frac{\dd Y}{\dd G}\right|_{M/P,BP}=\frac{1}{1-c(1-t)+(k/h)(b+\kappa)}.
\label{eq:flex_multiplier_m}
\end{equation}
As \(\kappa\) rises, the denominator rises, and the multiplier falls:
\begin{equation}
\pder{}{\kappa}\left[1-c(1-t)+\frac{k}{h}(b+\kappa)\right]=\frac{k}{h}>0.
\end{equation}
The classical conclusion that fiscal policy is weakened under high capital mobility and flexible exchange rates follows from this denominator. But again the result is derived with respect to \(G\), not with respect to a vector of heterogeneous instruments.

\subsection*{The aggregation problem inside the canonical derivation}

The scalar \(G\) appears in \eqref{eq:is_raw} as if one additional unit of any public-spending component had the same effect on demand. If the actual fiscal vector is \(\gvec\), the right-hand side should be written as
\begin{equation}
A-bi+\eta e+\Omega(\gvec),
\end{equation}
where \(\Omega\) maps fiscal instruments into effective demand and supply channels. The scalar derivation is valid only if
\begin{equation}
\Omega(\gvec)=\omega(\ones'\gvec)
\end{equation}
locally. The next section proves the exact condition for this representation.

\section{Fiscal aggregation: geometry, propositions and bias}

Let output be
\begin{equation}
Y=F(\gvec,X),
\end{equation}
where \(\gvec\in\R^n\) is the fiscal vector and \(X\in\R^k\) collects all non-fiscal states: money supply, exchange rate, debt, risk, external demand, productivity, institutional quality and so on. The scalar fiscal aggregate is
\begin{equation}
G=\ones'\gvec.
\end{equation}
The aggregation map \(P:\R^n\to\R\) is \(P\gvec=\ones'\gvec\). Its rank is one:
\begin{equation}
\rank(P)=1.
\end{equation}
By rank-nullity,
\begin{equation}
\dim\kernel(P)=n-1.
\end{equation}
Therefore, when \(n>1\), there are infinitely many fiscal recompositions that leave \(G\) unchanged. For \(n=3\), with \(\gvec=(G_C,G_I,TR)'\), two basis directions in the null space are
\begin{equation}
v_1=(-1,1,0)',\qquad v_2=(-1,0,1)'.
\end{equation}
Both satisfy \(\ones'v_1=\ones'v_2=0\). They are invisible to scalar \(G\), but not necessarily invisible to output.

\begin{proposition}[Local sufficiency of fiscal aggregation]
Let \(F:\R^n\times\R^k\to\R\) be continuously differentiable. A scalar aggregate \(G=\ones'\gvec\) is locally sufficient for the first-order fiscal effect on \(Y=F(\gvec,X)\), in the sense that every zero-sum fiscal recomposition has zero first-order effect on output, if and only if there exists a scalar \(\lambda(\gvec,X)\) such that
\begin{equation}
\nabla_{\gvec}F(\gvec,X)=\lambda(\gvec,X)\ones.
\label{eq:gradient_condition}
\end{equation}
\end{proposition}

\begin{proof}
A recomposition that leaves \(G\) unchanged satisfies \(\ones'\dd\gvec=0\). The first-order effect on output is
\begin{equation}
\dd Y=\nabla_{\gvec}F(\gvec,X)'\dd\gvec.
\end{equation}
This is zero for every \(\dd\gvec\in\kernel(P)\) if and only if \(\nabla_{\gvec}F\) lies in the orthogonal complement of \(\kernel(P)\). Since \(P\) has row space \(\operatorname{span}\{\ones\}\), the orthogonal complement of \(\kernel(P)\) is \(\operatorname{span}\{\ones\}\). Hence \(\nabla_{\gvec}F=\lambda\ones\). Conversely, if \(\nabla_{\gvec}F=\lambda\ones\), then \(\dd Y=\lambda\ones'\dd\gvec=0\) for every zero-sum recomposition. This proves the result.
\end{proof}

\begin{proposition}[Global sufficiency on connected level sets]
Fix \(X\). Suppose the fiscal domain \(\mathcal{D}\subseteq\R^n\) is convex and \(F(\cdot,X)\) is continuously differentiable. There exists a scalar function \(f\) such that
\begin{equation}
F(\gvec,X)=f(\ones'\gvec,X)
\end{equation}
for all \(\gvec\in\mathcal{D}\) if and only if \(\nabla_{\gvec}F(\gvec,X)=\lambda(\gvec,X)\ones\) throughout \(\mathcal{D}\).
\end{proposition}

\begin{proof}
If \(F(\gvec,X)=f(\ones'\gvec,X)\), then the chain rule gives
\begin{equation}
\nabla_{\gvec}F(\gvec,X)=f_G(\ones'\gvec,X)\ones,
\end{equation}
so the gradient condition holds with \(\lambda=f_G\). Conversely, suppose \(\nabla_{\gvec}F=\lambda\ones\) everywhere. Let \(\gvec_a\) and \(\gvec_b\) satisfy \(\ones'\gvec_a=\ones'\gvec_b\). Because \(\mathcal{D}\) is convex, the line path \(\gamma(\theta)=\gvec_a+\theta(\gvec_b-\gvec_a)\) lies in \(\mathcal{D}\). Along this path,
\begin{equation}
\frac{\dd}{\dd\theta}F(\gamma(\theta),X)=\nabla_{\gvec}F(\gamma(\theta),X)'(\gvec_b-\gvec_a)=\lambda(\gamma(\theta),X)\ones'(\gvec_b-\gvec_a)=0.
\end{equation}
Hence \(F(\gvec_a,X)=F(\gvec_b,X)\). The function is therefore constant on each level set of \(\ones'\gvec\), so define \(f(G,X)\) as that common value. This proves global sufficiency on connected level sets.
\end{proof}

\begin{remark}[Local versus global aggregation]
The local theorem concerns first-order effects around a given fiscal package. The global theorem is stronger: it requires the gradient condition throughout a connected fiscal domain. In applications, local sufficiency is enough to interpret marginal multipliers, while global sufficiency is needed to collapse all fiscal packages with the same total \(G\) into one scalar object.
\end{remark}

\begin{corollary}[Marginal-effect inequality destroys sufficiency]
If there exist instruments \(j,k\) such that
\begin{equation}
\pder{F}{g_j}\neq\pder{F}{g_k},
\end{equation}
then \(G=\ones'\gvec\) is not locally sufficient for the fiscal effect on output.
\end{corollary}

For a linear fiscal map,
\begin{equation}
F(\gvec,X)=H(X)+a_CG_C+a_IG_I+a_TTR,
\end{equation}
the gradient is
\begin{equation}
\nabla_{\gvec}F=(a_C,a_I,a_T)'.
\end{equation}
The zero-sum recomposition \(v_1=(-1,1,0)'\) changes output by
\begin{equation}
\nabla F'v_1=-a_C+a_I,
\end{equation}
while \(v_2=(-1,0,1)'\) changes output by
\begin{equation}
\nabla F'v_2=-a_C+a_T.
\end{equation}
Thus, the scalar aggregate hides economically relevant changes whenever \(a_C\), \(a_I\) and \(a_T\) differ.

\begin{proposition}[Composition-weighted multiplier]
Let a marginal fiscal impulse be allocated according to \(\dd\gvec=\wvec\dd G\), where \(\ones'\wvec=1\). Then the scalar multiplier generated by that impulse is
\begin{equation}
\frac{\dd Y}{\dd G}=\nabla_{\gvec}F(\gvec,X)'\wvec.
\label{eq:weighted_multiplier}
\end{equation}
It is invariant to \(\wvec\) if and only if \(\nabla_{\gvec}F=\lambda\ones\).
\end{proposition}

\begin{proof}
Substitute \(\dd\gvec=\wvec\dd G\) into the total differential:
\begin{equation}
\dd Y=\nabla_{\gvec}F'\dd\gvec=\nabla_{\gvec}F'\wvec\dd G.
\end{equation}
Dividing by \(\dd G\) gives \eqref{eq:weighted_multiplier}. If \(\nabla_{\gvec}F=\lambda\ones\), then \(\nabla_{\gvec}F'\wvec=\lambda\ones'\wvec=\lambda\), independent of \(\wvec\). Conversely, if \(\nabla_{\gvec}F'\wvec\) is invariant over all feasible \(\wvec\) on the simplex, all components of \(\nabla_{\gvec}F\) must be equal.
\end{proof}

\begin{corollary}[Aggregation bias]
Suppose a scalar-\(G\) model uses multiplier \(\bar\lambda\), while the true marginal impulse has allocation \(\wvec\). The first-order aggregation error is
\begin{equation}
\mc{B}(\wvec)=\left[\nabla_{\gvec}F(\gvec,X)'\wvec-\bar\lambda\right]\dd G.
\label{eq:aggregation_bias}
\end{equation}
If \(\bar\lambda\) is estimated under allocation \(\bar\wvec\), then
\begin{equation}
\bar\lambda=\nabla_{\gvec}F'\bar\wvec,
\end{equation}
and the composition-transfer error is
\begin{equation}
\mc{B}(\wvec,\bar\wvec)=\nabla_{\gvec}F'(\wvec-\bar\wvec)\dd G.
\end{equation}
\end{corollary}

This corollary explains why multiplier estimates must report the composition of the fiscal impulse. A multiplier estimated from a military procurement shock, a cash-transfer program or an infrastructure program need not generalize to a package with a different \(\wvec\).

\section{A fiscal-composition extension of IS--LM--BP}

This section writes a minimal model rich enough to distinguish current spending, public investment and transfers. The point is not to build a full DSGE or HANK model. The point is to make explicit the channels suppressed by a scalar \(G\).

\subsection*{Heterogeneous consumption and transfers}

Let households be indexed by \(q=1,\ldots,Q\). Household-group consumption is
\begin{equation}
C_{q,t}=C_{0q}+c_q(Y_{q,t}-T_{q,t}+TR_{q,t}),
\end{equation}
with \(0\leq c_q\leq 1\). Aggregate consumption is
\begin{equation}
C_t=\sum_{q=1}^Q C_{q,t}.
\label{eq:heterogeneous_consumption}
\end{equation}
If \(Y_{q,t}=s_qY_t\), \(\sum_qs_q=1\), then the income feedback term is
\begin{equation}
\sum_{q=1}^Qc_qs_qY_t=\bar c_sY_t,
\qquad
\bar c_s=\sum_{q=1}^Qc_qs_q.
\end{equation}
A transfer to group \(q\) raises consumption demand by \(c_q\dd TR_q\), but only the domestically absorbed part contributes to domestic output. If \(\mu_q\) is the import leakage of the marginal consumption basket, the direct domestic absorption coefficient is
\begin{equation}
\omega_{TR,q}=c_q(1-\mu_q).
\end{equation}
This is why transfers targeted to high-MPC liquidity-constrained households can have large impact effects, while transfers to low-MPC households can have smaller demand effects.

\subsection*{Current spending and public investment absorption}

Let current purchases have import leakage \(\mu_C\). The domestic absorption of current spending is
\begin{equation}
\omega_CG_{C,t}=(1-\mu_C)G_{C,t}.
\end{equation}
Let public investment project \(r\) have import leakage \(\mu_{I,r}\). Its direct contemporaneous domestic absorption is
\begin{equation}
\omega_{I,r}G_{I,r,t}=(1-\mu_{I,r})G_{I,r,t}.
\end{equation}
Current spending and public investment can therefore have different impact multipliers even before introducing a capital-stock channel.

\subsection*{Public-capital accumulation}

Public investment has an intertemporal channel. Let
\begin{equation}
K_{g,r,t+1}=(1-\delta_r)K_{g,r,t}+\phi_rG_{I,r,t},
\label{eq:kg_law}
\end{equation}
where \(\delta_r\in[0,1)\) is depreciation and \(\phi_r\in[0,1]\) is implementation efficiency. The parameter \(\phi_r\) is crucial: a budgetary unit of investment spending is not automatically a unit of productive public capital. Delays, corruption, design failure, overpricing and poor maintenance reduce \(\phi_r\).

Let potential output be
\begin{equation}
Y_t^*=A_tK_t^\alpha L_t^{1-\alpha}\prod_{r=1}^RK_{g,r,t}^{\psi_r}.
\label{eq:production}
\end{equation}
Taking logs,
\begin{equation}
\log Y_t^*=\log A_t+\alpha\log K_t+(1-\alpha)\log L_t+\sum_{r=1}^R\psi_r\log K_{g,r,t}.
\end{equation}
Thus,
\begin{equation}
\pder{Y_t^*}{K_{g,r,t}}=\psi_r\frac{Y_t^*}{K_{g,r,t}}.
\end{equation}
A marginal investment at \(t\) changes the public-capital stock at \(t+s\) by
\begin{equation}
\pder{K_{g,r,t+s}}{G_{I,r,t}}=\phi_r(1-\delta_r)^{s-1},\qquad s\geq 1.
\end{equation}
Therefore,
\begin{equation}
\pder{Y_{t+s}^*}{G_{I,r,t}}=\psi_r\frac{Y_{t+s}^*}{K_{g,r,t+s}}\phi_r(1-\delta_r)^{s-1}.
\label{eq:future_effect}
\end{equation}
This is the formal reason public investment must be evaluated over time rather than only through its impact effect.

\subsection*{External balance and risk}

The external block can be written as
\begin{equation}
BP_t=NX_t+CF_t+\Delta R_t=0.
\end{equation}
A reduced-form capital-flow equation is
\begin{equation}
CF_t=\varphi\left(i_t-i_t^*-\E_t\Delta s_{t+1}-\rho_t\right)\Omega_t+\varepsilon_t,
\label{eq:capital_flow}
\end{equation}
where \(s_t\) is the log exchange rate, \(\rho_t\) is a risk premium, \(\Omega_t\) captures global liquidity and \(\varepsilon_t\) is an external shock. A simple risk-premium equation is
\begin{equation}
\rho_t=\rho_0+\rho_1\frac{B_t}{Y_t}+\rho_2\frac{B_t^{ext}}{X_t}+\rho_3\sigma_t.
\end{equation}
This specification is reduced-form, but it captures a point absent from the deterministic BP schedule: debt, risk and liquidity can shift external financing conditions.

\subsection*{Debt dynamics}

Let government debt evolve as
\begin{equation}
B_{t+1}=(1+i_t)B_t+G_{C,t}+\sum_rG_{I,r,t}+\sum_qTR_{q,t}-T_t.
\label{eq:debt}
\end{equation}
If \(d_t=B_t/Y_t\), then
\begin{equation}
d_{t+1}=\frac{B_{t+1}}{Y_{t+1}}.
\end{equation}
Differentiating with respect to fiscal instrument \(g_j\),
\begin{equation}
\pder{d_{t+1}}{g_j}=\frac{Y_{t+1}\pder{B_{t+1}}{g_j}-B_{t+1}\pder{Y_{t+1}}{g_j}}{Y_{t+1}^2}.
\label{eq:debt_ratio_derivative}
\end{equation}
A fiscal expansion can reduce the debt ratio only if its effect on the denominator is sufficiently large relative to its effect on the numerator:
\begin{equation}
B_{t+1}\pder{Y_{t+1}}{g_j}>Y_{t+1}\pder{B_{t+1}}{g_j}.
\end{equation}
This condition is instrument-specific. It cannot be evaluated from scalar \(G\) alone.

\section{Instrument-specific multipliers}

This section derives the instrument-specific multipliers in the stylized model. Let the reduced-form denominator of the contemporaneous demand response be
\begin{equation}
D_t=1-\bar c_{s,t}+m_t+\omega_{f,t}+\omega_{\rho,t},
\label{eq:D}
\end{equation}
where \(m_t\) is openness/import leakage in the macro denominator, \(\omega_{f,t}\) a financial crowding-out term and \(\omega_{\rho,t}\) a risk term. The denominator is a reduced-form way of collecting the IS--LM--BP channels derived in Section 3.

\subsection*{Current spending}

The impact multiplier of current spending is
\begin{equation}
\mc{M}_{C,t}^{0}=\pder{Y_t}{G_{C,t}}=\frac{1-\mu_C}{D_t}.
\label{eq:current_multiplier}
\end{equation}
This expression shows that even current spending need not have multiplier one. Its domestic effect falls when import leakage \(\mu_C\) rises or when the macro denominator \(D_t\) rises because of openness, interest-rate feedback, risk or weak fiscal position.

\subsection*{Transfers}

The impact multiplier of a transfer to group \(q\) is
\begin{equation}
\mc{M}_{TR,q,t}^{0}=\pder{Y_t}{TR_{q,t}}=\frac{c_q(1-\mu_q)}{D_t}.
\label{eq:transfer_multiplier}
\end{equation}
Thus a transfer has no single multiplier independent of the recipient group. It rises with the marginal propensity to consume \(c_q\) and falls with the import leakage \(\mu_q\). This expression clarifies why transfer programs can dominate impact effects when targeted to high-MPC, low-import-leakage households, even if they do not generate public capital.

\subsection*{Public investment}

For project \(r\), the impact component is
\begin{equation}
\mc{M}_{I,r,t}^{0}=\frac{1-\mu_{I,r}}{D_t}.
\end{equation}
The present-value multiplier adds future public-capital effects:
\begin{equation}
\mc{M}_{I,r,t}^{PV}=\frac{1-\mu_{I,r}}{D_t}+\sum_{s=1}^S\beta^s\frac{\zeta_{r,t+s}\phi_r(1-\delta_r)^{s-1}+\psi_r\frac{Y_{t+s}^*}{K_{g,r,t+s}}\phi_r(1-\delta_r)^{s-1}}{D_{t+s}}-\mc{C}_{r,t}.
\label{eq:investment_pv}
\end{equation}
The term \(\zeta\) captures direct complementarity between public capital and realized output. The term involving \(\psi_r\) captures the potential-output channel. \(\mc{C}_{r,t}\) collects financing, risk, delay, inflationary and external costs. Under the simplifying approximation \(D_{t+s}=D\), \(Y_{t+s}^*/K_{g,r,t+s}=\bar y_k\), \(\zeta_{r,t+s}=\zeta\), and \(\mc{C}_{r,t}=0\), the future channel becomes a finite geometric sum:
\begin{equation}
\sum_{s=1}^S\beta^s\frac{(\zeta+\psi_r\bar y_k)\phi_r(1-\delta_r)^{s-1}}{D}=\frac{(\zeta+\psi_r\bar y_k)\phi_r}{D}\beta\frac{1-[\beta(1-\delta_r)]^S}{1-\beta(1-\delta_r)}.
\label{eq:geometric_investment}
\end{equation}
Equation \eqref{eq:geometric_investment} explains why the value of public investment rises with implementation efficiency \(\phi_r\), public-capital productivity \(\psi_r\), persistence \(1-\delta_r\), the discount factor \(\beta\), and the horizon \(S\).

\subsection*{When does investment dominate current spending?}

Investment dominates current spending in present value when
\begin{equation}
\mc{M}_{I,r,t}^{PV}>\mc{M}_{C,t}^{0}.
\end{equation}
Using \eqref{eq:current_multiplier} and \eqref{eq:investment_pv}, this requires
\begin{equation}
\sum_{s=1}^S\beta^s\frac{\left[\zeta_{r,t+s}+\psi_r\frac{Y_{t+s}^*}{K_{g,r,t+s}}\right]\phi_r(1-\delta_r)^{s-1}}{D_{t+s}}>\frac{\mu_{I,r}-\mu_C}{D_t}+\mc{C}_{r,t}.
\label{eq:investment_condition}
\end{equation}
This inequality is central. It shows that public investment is not superior by definition. It dominates only if the present value of productivity and complementarity effects exceeds any additional import leakage and cost terms. This reconciles the meta-analytic tendency for public investment to have large multipliers with the practical fact that poorly implemented or import-intensive projects can perform badly.

\subsection*{A social objective with transfers and investment}

If policy design includes welfare as well as output, a fiscal planner may solve
\begin{equation}
\max_{\{G_C,G_{I,r},TR_q\}}\sum_{s=0}^S\beta^sY_{t+s}+\lambda_W W_t
\end{equation}
subject to budget, external and inflation constraints. For a transfer to group \(q\), the first-order condition includes both the output effect and the welfare effect:
\begin{equation}
\sum_{s=0}^S\beta^s\pder{Y_{t+s}}{TR_{q,t}}+\lambda_W\pder{W_t}{TR_{q,t}}=\lambda_B+\lambda_F\mu_q+\lambda_\pi\pder{\pi_t}{TR_{q,t}}.
\end{equation}
For investment project \(r\), the condition is
\begin{equation}
\sum_{s=0}^S\beta^s\pder{Y_{t+s}}{G_{I,r,t}}=\lambda_B+\lambda_F\mu_{I,r}+\lambda_\pi\pder{\pi_t}{G_{I,r,t}}+\lambda_R\pder{\rho_t}{G_{I,r,t}}.
\end{equation}
Thus, policy choice is not simply a comparison of scalar multipliers. It is a constrained allocation problem.

\section{Computational implementation and why each test was used}

The numerical model is deliberately transparent. The replication code is included in the submission package and is also available as an external Colab notebook at \url{https://colab.research.google.com/drive/1B24i_OYeVfsyQGasdQujvq9_hvuOgUut?usp=sharing}. Baseline output and public capital are normalized to 100. Each experiment applies the same one-period fiscal impulse and changes only the composition. The denominator is
\begin{equation}
D=1-\bar c+m+\omega_f+\omega_\rho+\omega_d\max(d_0-\bar d,0),
\end{equation}
where \(\omega_f\) is a financial penalty and \(\omega_\rho+\omega_d\max(d_0-\bar d,0)\) captures risk and debt fragility. Output deviations are computed as
\begin{equation}
\Delta Y_t=\frac{\text{demand}_t}{D}+\frac{\zeta\Delta K_{g,t}+\Delta Y_t^*}{D}-\text{risk drag}_t.
\end{equation}
Potential-output deviation is
\begin{equation}
\Delta Y_t^*=\psi\frac{Y_0}{K_{g,0}}\Delta K_{g,t}.
\end{equation}
The public-capital law is
\begin{equation}
\Delta K_{g,t+1}=(1-\delta_g)\Delta K_{g,t}+\phi G_{I,t}.
\end{equation}
The external-balance proxy is
\begin{equation}
NX_t=-\text{fiscal imports}_t-n_x\Delta Y_t+\chi\Delta K_{g,t}.
\end{equation}
Debt evolves according to
\begin{equation}
\Delta B_{t+1}=(1+r)\Delta B_t+\text{fiscal cost}_t-\tau\Delta Y_t.
\end{equation}
Inflation pressure is measured as
\begin{equation}
\pi_t=\lambda_\pi(\Delta Y_t-\Delta Y_t^*).
\end{equation}

\begin{table}[!htbp]\centering
\caption{Baseline parameters and adversarial ranges.}\label{tab:parameters}
\begin{threeparttable}
\scriptsize
\begin{tabular}{p{0.13\textwidth}p{0.27\textwidth}p{0.12\textwidth}p{0.16\textwidth}p{0.24\textwidth}}
\toprule
Parameter & Interpretation & Baseline & Range & Rationale\\
\midrule
$\beta$ & Discount factor & 0.96 & 0.90--0.985 & Medium-run present-value horizon; varied in robustness checks.\\
$\bar c$ & Aggregate consumption feedback & 0.68 & 0.48--0.88 & Captures reduced-form Keynesian demand feedback.\\
$m$ & Openness/import denominator & 0.22 & 0.02--0.55 & Open-economy leakage; motivated by evidence that openness lowers multipliers.\\
$\omega_f$ & Financial penalty & 0.18 & 0--0.70 & Reduced-form interest-rate/crowding-out channel.\\
$\omega_\rho$ & Risk penalty & 0.05 & 0--0.40 & Reduced-form sovereign/external-finance risk channel.\\
$\mu_C$ & Import leakage, current spending & 0.22 & 0.02--0.70 & Allows domestic absorption to differ across instruments.\\
$\mu_I$ & Import leakage, public investment & 0.28 & 0.02--0.90 & Captures imported capital-goods content and project dependence.\\
$\mu_p,\mu_r$ & Import leakage, transfers & 0.18, 0.36 & 0.02--0.80 & Allows consumption baskets to differ by group.\\
$c_p,c_r$ & MPCs of poorer/richer recipients & 0.90, 0.45 & 0.65--0.98; 0.15--0.70 & Reflects heterogeneous liquidity and marginal propensities to consume.\\
$\phi$ & Implementation efficiency & 0.75 & 0--1 & Share of investment spending converted into productive public capital.\\
$\psi$ & Public-capital output elasticity & 0.12 & 0--0.25 & Consistent with positive but heterogeneous public-capital productivity literature.\\
$\delta_g$ & Public-capital depreciation & 0.07 & 0.02--0.18 & Controls persistence of public-capital effects.\\
$\zeta$ & Direct public-capital demand/productivity channel & 0.08 & 0--0.20 & Reduced-form complementarity between public capital and output.\\
$\chi$ & External-balance effect of public capital & 0.02 & -0.02--0.08 & Captures infrastructure effects on trade costs/net exports.\\
$\tau$ & Tax feedback & 0.18 & 0.08--0.32 & Revenue response to output deviations.\\
$d_0$ & Initial debt ratio & 0.60 & 0.15--1.50 & Fiscal-space/risk sensitivity channel.\\
\bottomrule
\end{tabular}
\begin{tablenotes}\footnotesize\item Notes: Ranges are illustrative priors used for computational consistency checks, not estimates. The purpose is to stress-test the aggregation claim over broad admissible configurations.\end{tablenotes}
\end{threeparttable}
\end{table}

The tests were chosen to match the theoretical claims. The symbolic tests verify the aggregation theorem, null-space recompositions, a nonlinear counterexample and the finite geometric public-capital formula. Finite-difference tests check that the numerical impact derivatives equal the analytical derivatives. Accounting-identity tests verify that the simulated debt, capital and external-balance equations are internally consistent. Adversarial tests construct cases where poor-household transfers dominate current spending in impact, cases where investment fails because import leakage is extreme and productivity is zero, and cases where investment dominates under high efficiency, high productivity and low import leakage. Sensitivity tests verify monotonicity in implementation efficiency, investment import leakage, debt and openness. Monte Carlo tests verify that the scalar-\(G\) ranking is not universal over broad parameter ranges. Extreme stress tests verify that the code remains numerically well behaved under wide parameter supports. Seed tests verify reproducibility.

\section{Results}

\begin{table}[!htbp]\centering
\caption{Baseline equal-aggregate fiscal impulse by composition.}\label{tab:baseline}
\begin{tabular}{lrr}
\toprule
Composition & Impact effect on $Y$ & $PV(Y)$\\
\midrule
Current spending & 5.0649 & 5.0548\\
Public investment & 4.6753 & 12.3784\\
Poor-household transfer & 4.7922 & 4.7820\\
Rich-household transfer & 1.8701 & 1.8586\\
Mixed package & 4.1006 & 6.0184\\
\bottomrule
\end{tabular}
\begin{tablenotes}\footnotesize\item Notes: All rows use the same total one-period fiscal impulse; only composition changes. Values are model units relative to normalized baseline output.\end{tablenotes}
\end{table}

Table \ref{tab:baseline} reports the central numerical implication. Every row uses the same aggregate one-period fiscal impulse. The scalar-\(G\) model predicts the same contemporaneous value for all rows when it ignores composition. The fiscal-composition model produces different impact effects and different present values. Public investment has a lower impact effect than current spending in the baseline because its direct import leakage is slightly higher, but it has a much larger present value because it accumulates public capital. A transfer to poorer households has a higher impact effect than a transfer to richer households because \(c_p(1-\mu_p)>c_r(1-\mu_r)\). The mixed package lies between the pure policies.

\begin{figure}[!htbp]
\centering
\includegraphics[width=0.92\textwidth]{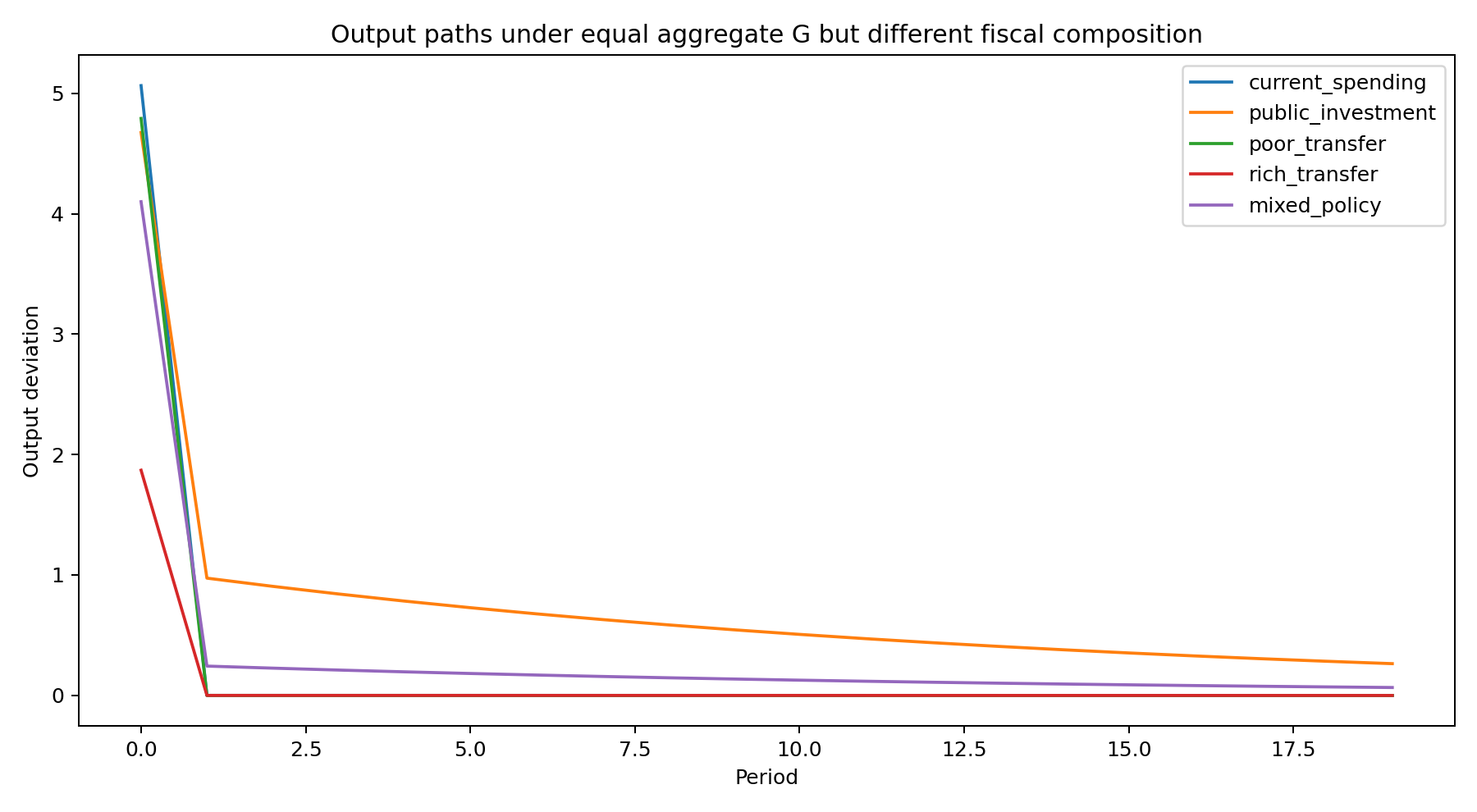}
\caption{Output paths under equal aggregate spending but different fiscal composition.}
\label{fig:paths}
\end{figure}

Figure \ref{fig:paths} shows the time profile. Current spending and transfers are largely impact policies in the stylized model. Public investment generates a smaller immediate impact but a persistent path through public capital. This is exactly the distinction hidden by scalar \(G\).

\begin{table}[!htbp]\centering
\caption{Monte Carlo and stress-test summary.}\label{tab:mcsummary}
\begin{tabular}{lr}
\toprule
Statistic & Value\\
\midrule
Monte Carlo draws & 3,000\\
Stress-test draws & 500\\
Share with composition-dependent PV(Y) & 1.0000\\
Investment wins share & 0.7883\\
Poor transfer wins impact share & 0.2277\\
Mean absolute scalar-G error & 2.1883\\
Current-spending winners & 426\\
Public-investment winners & 2,365\\
Poor-transfer winners & 203\\
Rich-transfer winners & 6.0000\\
Mixed-policy winners & 0.0000\\
Mean phi if investment wins & 0.5705\\
Mean phi otherwise & 0.2674\\
Mean psi if investment wins & 0.1310\\
Mean psi otherwise & 0.0965\\
Mean $\mu_I$ if investment wins & 0.4031\\
Mean $\mu_I$ otherwise & 0.6688\\
\bottomrule
\end{tabular}
\begin{tablenotes}\footnotesize\item Notes: Monte Carlo draws use broad admissible parameter intervals; stress tests use wider extreme intervals. The exercise is not an empirical calibration.\end{tablenotes}
\end{table}

Table \ref{tab:mcsummary} summarizes the Monte Carlo exercise. In 100 percent of draws, equal aggregate spending does not imply equal present-value output across compositions. Public investment wins in 78.83 percent of cases, but not in all cases. Current spending wins in 426 cases, poor-household transfers in 203 cases and rich-household transfers in 6 cases. The mean absolute error of the scalar-\(G\) prediction is 2.1883 model units. Investment wins under parameter draws with higher implementation efficiency and public-capital productivity and lower investment import leakage. This is the numerical counterpart of inequality \eqref{eq:investment_condition}.

\begin{figure}[!htbp]
\centering
\includegraphics[width=0.92\textwidth]{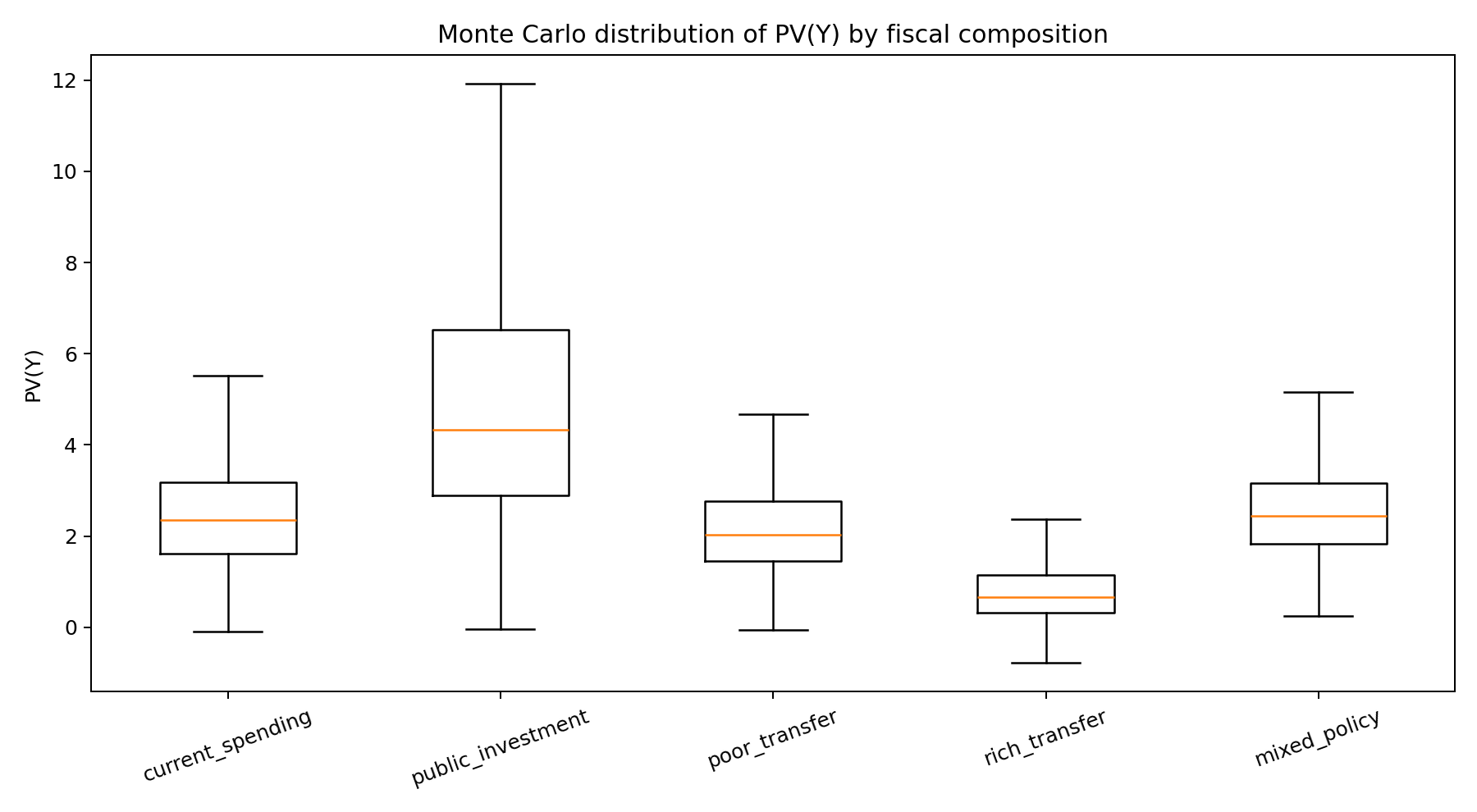}
\caption{Monte Carlo distribution of present-value output by fiscal composition.}
\label{fig:mcbox}
\end{figure}

Figure \ref{fig:mcbox} shows that the distributions overlap. This overlap is important. It means the paper does not claim a universal ranking. The model supports a conditional statement: public investment is often strong when efficient and productive, but it can be dominated by other instruments under adverse conditions.

\begin{figure}[!htbp]
\centering
\begin{subfigure}{0.48\textwidth}
\includegraphics[width=\textwidth]{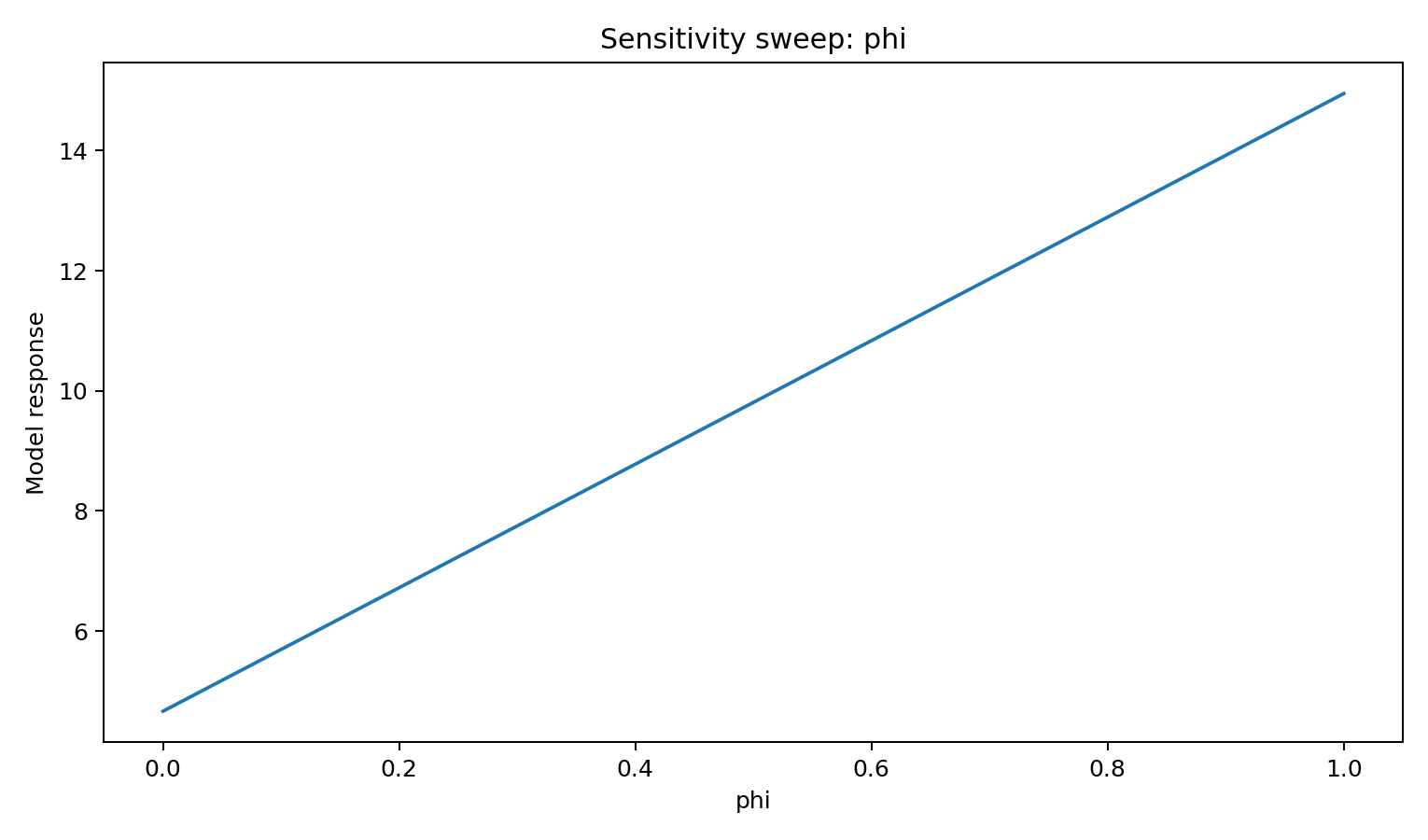}
\caption{Implementation efficiency \(\phi\).}
\end{subfigure}
\begin{subfigure}{0.48\textwidth}
\includegraphics[width=\textwidth]{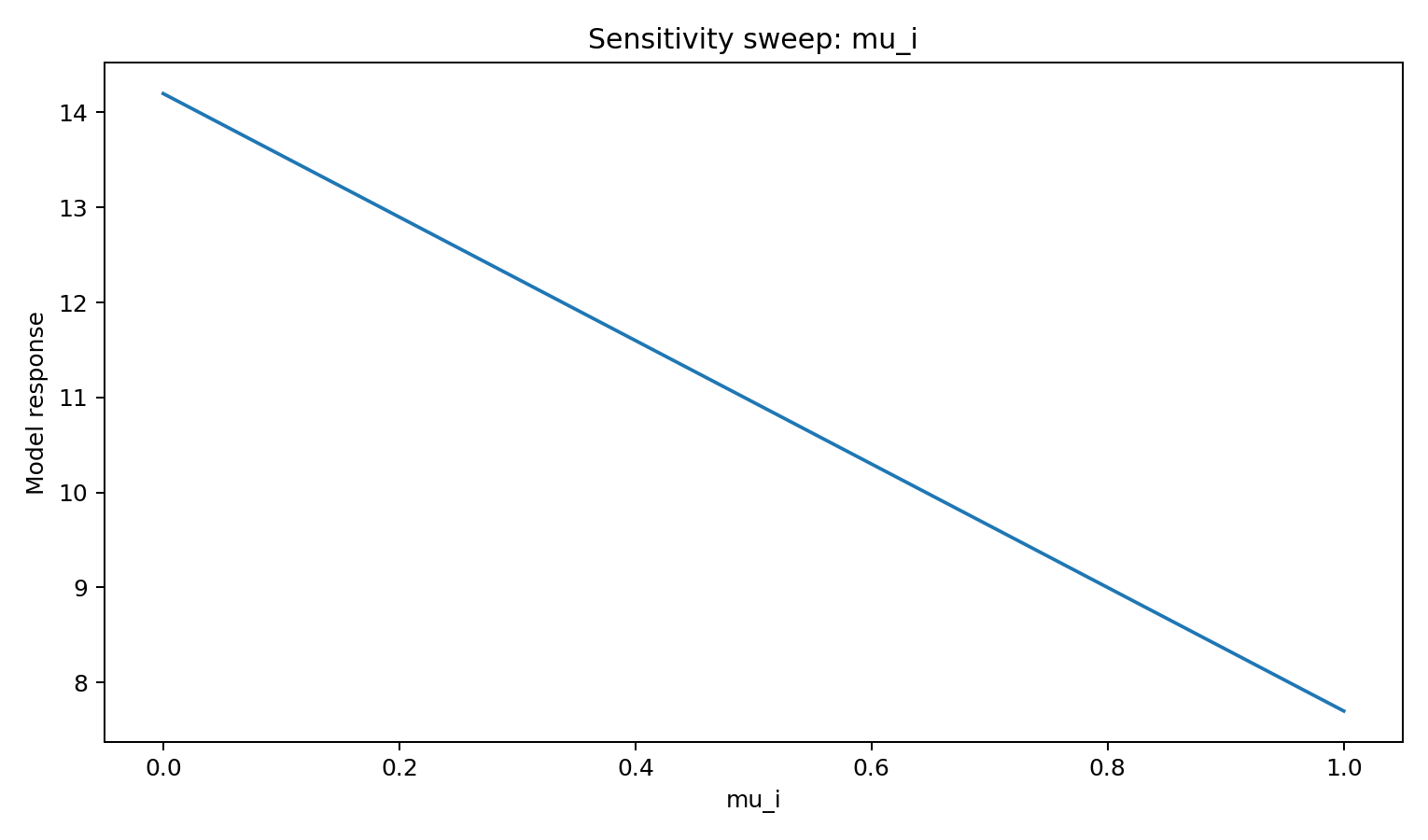}
\caption{Investment import leakage \(\mu_I\).}
\end{subfigure}
\begin{subfigure}{0.48\textwidth}
\includegraphics[width=\textwidth]{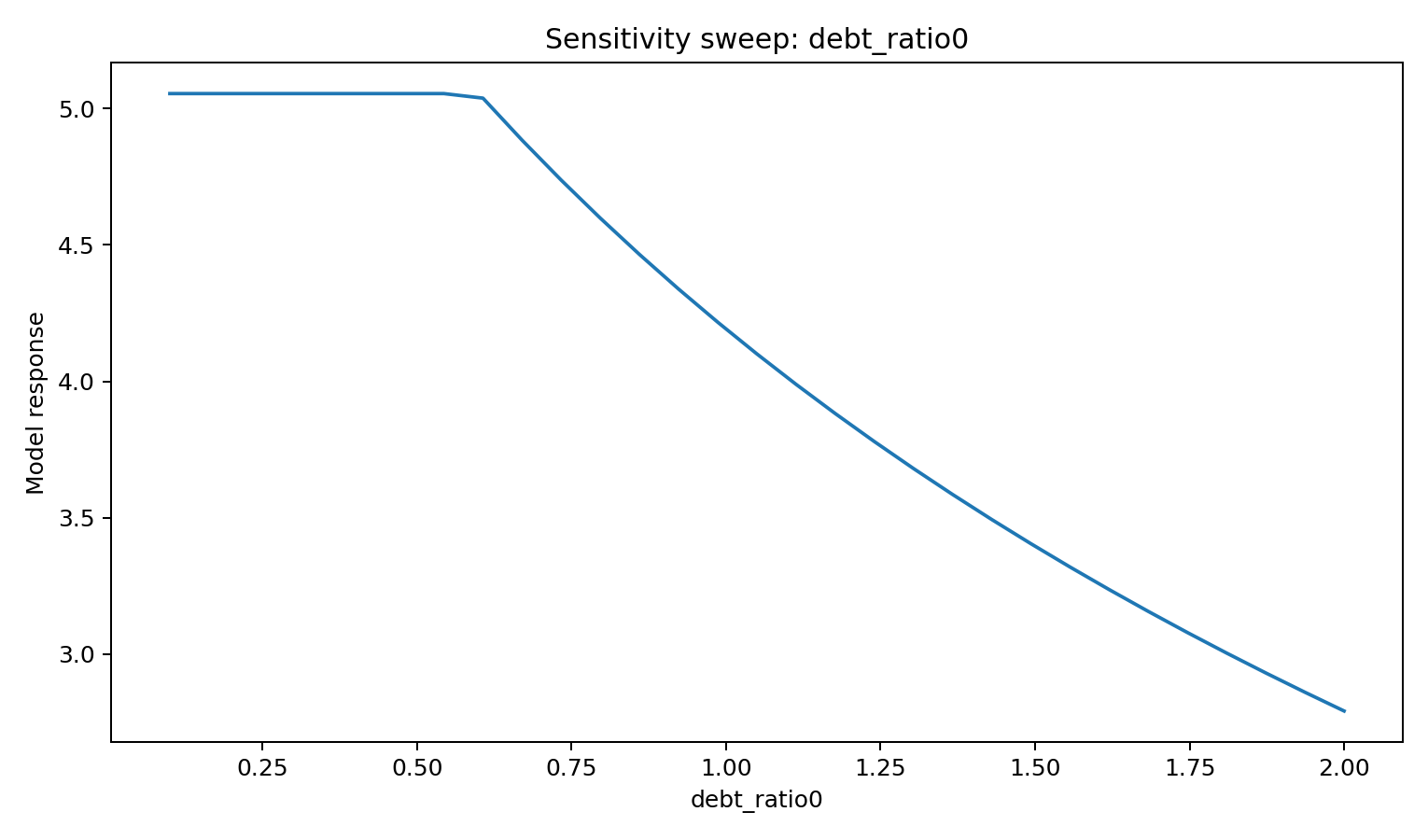}
\caption{Initial debt ratio.}
\end{subfigure}
\begin{subfigure}{0.48\textwidth}
\includegraphics[width=\textwidth]{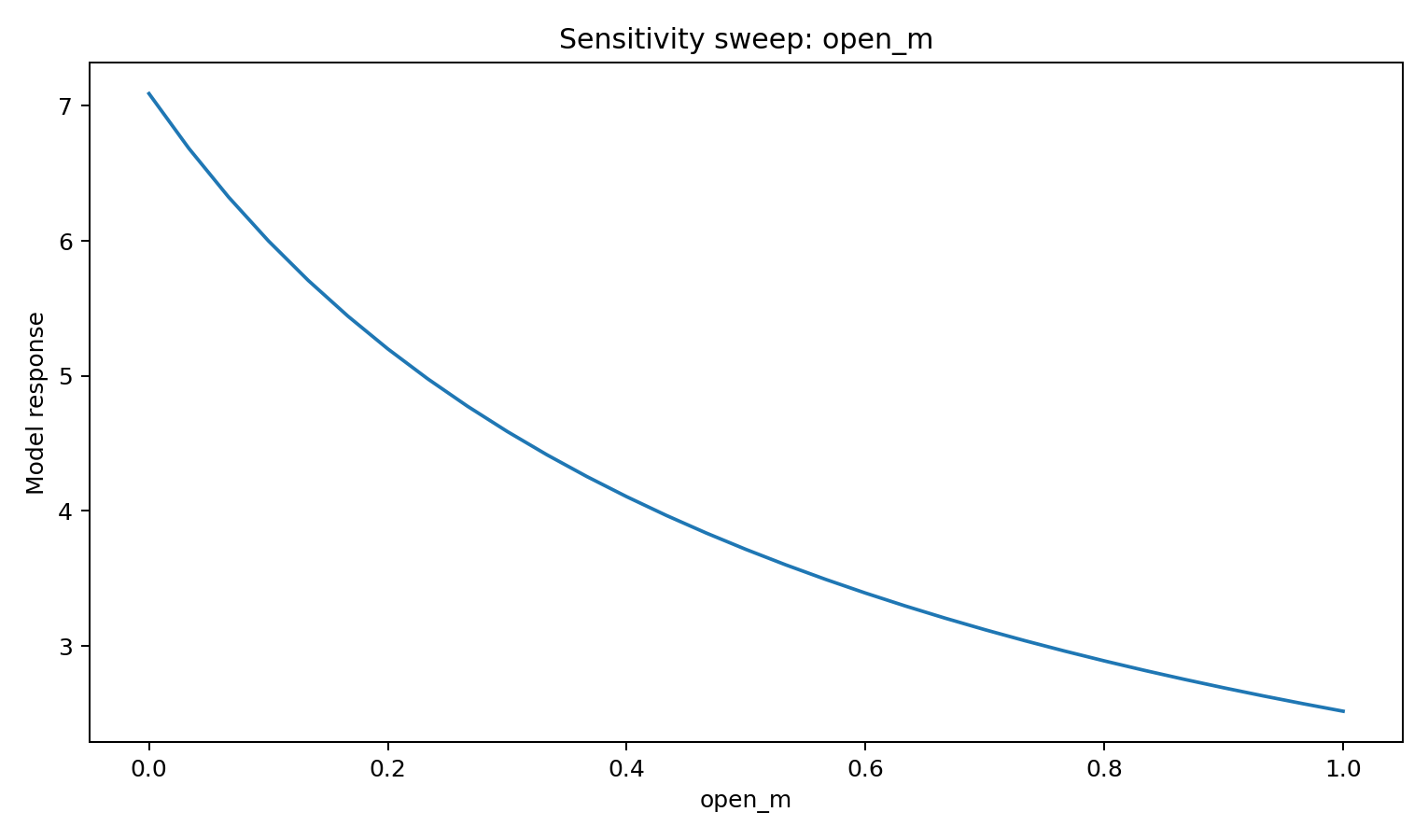}
\caption{Openness/import denominator.}
\end{subfigure}
\caption{Sensitivity sweeps.}
\label{fig:sweeps}
\end{figure}

Figure \ref{fig:sweeps} reports monotonicity checks. Investment present value rises with implementation efficiency and falls with investment import leakage. Current-spending effects fall with debt fragility and openness. These results are not empirical estimates; they verify that the computational model behaves according to the theoretical restrictions.

\begin{table}[!htbp]\centering
\caption{Validation battery by test family.}\label{tab:testfamilies}
\begin{tabular}{lrr}
\toprule
Test family & Tests & Passed\\
\midrule
Deterministic adversarial cases and identities & 21 & 21\\
Monte Carlo and stress testing & 10 & 10\\
Replication package generation & 1 & 1\\
Sensitivity sweeps & 4 & 4\\
Symbolic algebra and sufficiency & 6 & 6\\
\midrule
Total & 42 & 42\\
\bottomrule
\end{tabular}
\begin{tablenotes}\footnotesize\item Notes: Full test-level output is reported in Appendix~\ref{app:tests} and in the replication archive.\end{tablenotes}
\end{table}

The full battery passes 42/42 checks. The complete table is reported in Appendix \ref{app:tests}. The main point is that the computational exercise verifies internal consistency rather than empirical truth. It tests whether the formal claims are implemented correctly, whether the accounting equations hold, whether counterexamples behave as expected, and whether the scalar-\(G\) model fails under heterogeneous instruments.

\section{Discussion}

The main theoretical implication is that the fiscal multiplier is not a primitive. It is a derivative of output with respect to a particular fiscal direction. When policy is scalar, the direction is hidden. When policy is vector-valued, the direction is \(\wvec\), and the multiplier is \(\nabla F'\wvec\). This has direct implications for interpreting empirical multipliers. A multiplier estimated from an infrastructure shock, a transfer shock or a procurement shock should not be transferred to a different fiscal package unless the homogeneity condition is plausible.

The second implication concerns public investment. The model separates impact absorption from intertemporal productivity. Public investment can have a lower impact multiplier than current spending and still a higher present-value multiplier. It can also fail if implementation efficiency is low, import leakage high, depreciation rapid or risk costs large. Therefore, the correct statement is not that public investment always dominates. The correct statement is that public investment can dominate when the present value of its productivity and complementarity effects exceeds its additional costs and leakages.

The third implication concerns transfers. A transfer is not one instrument if recipients differ. Equation \eqref{eq:transfer_multiplier} shows that targeting matters through \(c_q\) and \(\mu_q\). Transfers to high-MPC households can be powerful impact stabilizers, while transfers to low-MPC or high-import-leakage households can have smaller domestic effects. A welfare objective can nevertheless justify transfers even when their output multiplier is not the highest, because welfare weights and poverty effects matter.

The fourth implication concerns the external constraint. The classical Mundell--Fleming model already shows that the exchange-rate regime and capital mobility condition fiscal effects. The extended model adds import content, risk and debt fragility. This matters especially in open economies and emerging markets. Fiscal design should therefore report the foreign-exchange content of the package and the financing environment.

A final implication concerns empirical reporting. A fiscal-multiplier estimate should report not only its point estimate and confidence interval, but also the composition vector of the shock that generated it. Without the composition vector, the estimate is incomplete for policy transfer. In the notation of this paper, the reported number is \(\nabla_{\gvec}F'\bar\wvec\), not \(\nabla_{\gvec}F'\wvec\) for every possible \(\wvec\). This is the direct policy content of the aggregation theorem.

\section{Limitations and threats to interpretation}

The paper makes a formal and computational argument, not an empirical one. The parameter ranges are illustrative and adversarial; they are not estimated for a specific country. The model is reduced-form. Risk, monetary response, external balance and fiscal financing are represented parsimoniously rather than derived from optimizing agents. Public-investment productivity is imposed through \(\phi\), \(\psi\), \(\delta_g\) and \(\zeta\), not estimated. The Monte Carlo shares are not probabilities about the world; they are outcomes under broad illustrative parameter ranges. The purpose of the numerical exercise is to test internal consistency and generate counterexamples to scalar sufficiency.

These limitations are deliberate. A more empirical paper would estimate instrument-specific multipliers with local projections, SVARs, narrative shocks, regional designs or project-level disbursement data. A more structural paper would embed the fiscal vector in a full DSGE or HANK model. This paper instead isolates one methodological point: scalar aggregation is valid only under a gradient homogeneity condition. That point applies regardless of the empirical strategy.

\section{Conclusion}

This paper develops a detailed mathematical critique of scalar fiscal aggregation in IS--LM--BP. The canonical model remains coherent under its assumptions, but the scalar public-spending variable \(G\) is not generally sufficient for fiscal policy analysis. The necessary and sufficient local condition is \(\nabla_{\gvec}Y=\lambda\ones\). When instruments have different marginal effects, the aggregate multiplier becomes \(\nabla_{\gvec}Y'\wvec\), a composition-weighted object.

The detailed derivations show how the canonical multipliers arise and why they depend on closure assumptions. The fiscal-composition model then shows how current spending, transfers and public investment generate distinct impact and present-value effects. Public investment has an intertemporal capital channel; transfers depend on recipient MPCs and import leakages; external constraints and debt fragility alter the denominator of fiscal effects. The computational exercise confirms the logic under symbolic, deterministic, adversarial, sensitivity and Monte Carlo tests.

The implication for research and policy is straightforward. Fiscal-multiplier estimates should report the composition of the fiscal impulse, the state of the economy, openness, debt, financing conditions, implementation efficiency and targeting. Without these dimensions, a scalar multiplier is not a structural policy parameter. It is an average over an implicit composition vector.

\section*{Declarations}

\textbf{Funding.} The author received no external funding for this research.\\
\textbf{Competing interests.} The author declares no competing interests.\\
\textbf{Ethics approval and consent to participate.} Not applicable.\\
\textbf{Consent for publication.} Not applicable.\\
\textbf{Data availability.} All generated numerical outputs used in the paper are included in the replication folder of the submission package.\\
\textbf{Code availability.} The complete one-cell Colab/Python replication script is included in \texttt{replication/colab\_one\_cell\_is\_lm\_bp\_extended\_validation.txt}. An external executable Colab version is available at \url{https://colab.research.google.com/drive/1B24i_OYeVfsyQGasdQujvq9_hvuOgUut?usp=sharing}.\\
\textbf{Author contribution.} Ricardo Alonzo Fern\'andez Salguero is the sole author and is responsible for conceptualization, formal analysis, computational design and manuscript preparation.

\FloatBarrier
\appendix

\section{Additional derivations}

\subsection{Second-order composition effects}

The local aggregation theorem is first-order. If \(F\) is twice differentiable, second-order effects are governed by the Hessian \(H_F\). For a zero-sum recomposition \(v\), the second-order output effect is
\begin{equation}
\Delta Y=\nabla F'v\epsilon+\frac{1}{2}\epsilon^2v'H_Fv+o(\epsilon^2).
\end{equation}
Even if \(\nabla F=\lambda\ones\), aggregation can fail at second order if \(v'H_Fv\neq 0\) for zero-sum \(v\). For example,
\begin{equation}
F(\gvec)=f(\ones'\gvec)+\gamma(G_C-G_I)^2
\end{equation}
has first-order composition effects except at \(G_C=G_I\), and second-order effects whenever \(\gamma\neq 0\). This motivates the nonlinear counterexample in the symbolic tests.

\subsection{Full derivation of transfer multipliers}

Let two groups have shares \(s_p\) and \(s_r\), MPCs \(c_p\) and \(c_r\), and transfer import leakages \(\mu_p\) and \(\mu_r\). Consumption is
\begin{equation}
C=C_{0p}+C_{0r}+c_p(s_pY-T_p+TR_p)+c_r(s_rY-T_r+TR_r).
\end{equation}
The income feedback is \((c_ps_p+c_rs_r)Y\). Define \(\bar c_s=c_ps_p+c_rs_r\). With denominator \(D=1-\bar c_s+m+\omega_f+\omega_\rho\), the domestic demand effect of transfers is
\begin{equation}
\dd Y=\frac{c_p(1-\mu_p)\dd TR_p+c_r(1-\mu_r)\dd TR_r}{D}.
\end{equation}
Hence
\begin{equation}
\pder{Y}{TR_p}=\frac{c_p(1-\mu_p)}{D},\qquad
\pder{Y}{TR_r}=\frac{c_r(1-\mu_r)}{D}.
\end{equation}
The transfer multiplier is therefore not a property of transfers alone. It is a property of targeting and household behavior.

\subsection{Present-value public-capital channel}

Given \(K_{g,t+1}=(1-\delta)K_{g,t}+\phi G_{I,t}\), the capital created at \(t+s\) by one unit of investment at \(t\) is \(\phi(1-\delta)^{s-1}\). If \(\partial Y^*/\partial K_g=\psi Y^*/K_g\), then
\begin{equation}
\sum_{s=1}^S\beta^s\pder{Y_{t+s}^*}{G_{I,t}}=\sum_{s=1}^S\beta^s\psi\frac{Y_{t+s}^*}{K_{g,t+s}}\phi(1-\delta)^{s-1}.
\end{equation}
Under the constant-ratio approximation \(Y_{t+s}^*/K_{g,t+s}=\bar y_k\),
\begin{equation}
\sum_{s=1}^S\beta^s\psi\bar y_k\phi(1-\delta)^{s-1}=\psi\bar y_k\phi\beta\frac{1-[\beta(1-\delta)]^S}{1-\beta(1-\delta)}.
\end{equation}
This is the expression verified by the symbolic test battery.

\subsection{Debt-ratio derivative}

Debt evolves as \(B_{t+1}=(1+i_t)B_t+g_t-T_t\). Let \(d_{t+1}=B_{t+1}/Y_{t+1}\). For instrument \(g_j\),
\begin{equation}
\pder{d_{t+1}}{g_j}=\frac{1}{Y_{t+1}}\pder{B_{t+1}}{g_j}-\frac{B_{t+1}}{Y_{t+1}^2}\pder{Y_{t+1}}{g_j}.
\end{equation}
This derivative can be positive or negative. A fiscal instrument can improve the debt ratio if its denominator effect on output is sufficiently large, but this condition is instrument-specific.

\section{Computational test battery}\label{app:tests}

\begin{longtable}{p{0.12\textwidth}p{0.37\textwidth}p{0.08\textwidth}p{0.32\textwidth}}
\caption{Complete computational consistency and adversarial test battery.}\label{tab:fulltests}\\
\toprule
Test & Name & Pass & Details\\
\midrule
\endfirsthead
\toprule Test & Name & Pass & Details\\ \midrule
\endhead
SYM-01 & Linear aggregation sufficiency & True & Requires a\_C=a\_I=a\_T.\\
SYM-02 & Zero-sum recompositions can affect output & True & Composition changes are invisible to scalar G.\\
SYM-03 & If F=f(sum G), zero-sum changes have zero first-order effect & True & Aggregation case verified.\\
SYM-04 & Nonlinear composition term violates aggregation & True & Counterexample verified.\\
SYM-05 & Flexible denominator rises with capital mobility & True & dD/dkappa=k/h.\\
SYM-06 & Public-capital PV equals finite geometric expression & True & Geometric PV verified.\\
DET-01 & Same aggregate G gives different PV(Y) & True & current\_spending=5.0548, public\_investment=12.3784, poor\_transfer=4.7820, rich\_transfer=1.8586, mixed\_policy=6.0184\\
DET-02 & Scalar-G model gives identical prediction & True & common=6.493506\\
DET-03 & Finite-difference derivatives match analytic derivatives & True & current\_spending: fd=1.01298701, an=1.01298701; public\_investment: fd=0.93506494, an=0.93506494; poor\_transfer: fd=0.95844156, an=0.95844156; rich\_transf...\\
DET-04 & Debt identity holds & True & All baseline periods checked.\\
DET-05 & Public-capital identity holds & True & All investment periods checked.\\
DET-06 & External-balance decomposition holds & True & NX identity checked.\\
DET-07 & Poor transfers can dominate current spending in impact & True & poor=6.2364; current=1.9481\\
DET-08 & Public investment does not universally dominate & True & best=current\_spending; investment PV=0.3124\\
DET-09 & Public investment dominates under favorable conditions & True & current\_spending=5.0548, public\_investment=36.8165, poor\_transfer=4.7820, rich\_transfer=1.8586, mixed\_policy=12.1282\\
DET-10 & Aggregation passes when marginal effects are homogeneous & True & 4.54545455, 4.54545455, 4.54545455, 4.54545455\\
DET-11 & Higher fiscal fragility lowers net fiscal effect & True & low debt PV=5.0379; high debt PV=2.4450\\
DET-12 & Higher openness lowers impact multiplier & True & low open=6.5000; high open=2.8889\\
DET-13 & Financial penalty lowers impact & True & low penalty=6.6102; high penalty=2.6174\\
DET-14 & Risk penalty lowers impact & True & low risk=5.4167; high risk=2.5658\\
DET-15 & Flexible-rate multiplier falls with capital mobility & True & M\_low=1.6304; M\_high=0.1123\\
DET-16 & Shock-size linearity holds in linearized case & True & small=2.025974; big=4.051948\\
DET-17 & Productive investment gains with longer horizon & True & H5=11.9112; H20=23.9886\\
DET-18 & Higher beta raises PV of future-oriented investment & True & beta low=13.1980; beta high=22.2274\\
DET-19 & Restricted closed-form difference matches simulation & True & closed=12.01412451; sim=12.01412451\\
DET-20 & Deterministic reproducibility & True & A=12.3783792471; B=12.3783792471\\
DET-21 & Baseline trajectories are finite & True & All arrays checked.\\
MC-01 & Same G almost never implies equal PV(Y) & True & 100.0000\% non-identical.\\
MC-02 & No universal fiscal ranking & True & current\_spending=426, public\_investment=2365, poor\_transfer=203, rich\_transfer=6, mixed\_policy=0\\
MC-03 & Investment wins sometimes but not always & True & investment win share=78.83\%\\
MC-04 & Poor transfers compete strongly in impact & True & poor-transfer best-impact share=22.77\%\\
MC-05 & Investment winners have higher phi/psi & True & phi 0.571$>$0.267; psi 0.131$>$0.097\\
MC-06 & Investment winners have lower import leakage & True & mu\_i 0.403$<$0.669\\
MC-07 & Scalar-G prediction has material composition error & True & MAE=2.1883\\
MC-08 & Monte Carlo outputs are finite & True & N=3000\\
MC-09 & Extreme stress tests remain finite & True & 500 extreme draws checked.\\
MC-10 & Monte Carlo reproducibility under fixed seed & True & max diff=0.00e+00\\
SENS-01 & Investment PV is non-decreasing in phi & True & phi sweep\\
SENS-02 & Investment PV is non-increasing in mu\_i & True & mu\_i sweep\\
SENS-03 & Current-spending PV is non-increasing in debt ratio & True & debt sweep\\
SENS-04 & Current-spending impact is non-increasing in openness & True & openness sweep\\
OUT-01 & Output ZIP package generated & True & ZIP package generated\\
\bottomrule
\end{longtable}


\begin{thebibliography}{99}

\bibitem[Abiad et~al.(2016)]{abiad2016}
Abiad, A., Furceri, D., and Topalova, P. (2016).
The macroeconomic effects of public investment: Evidence from advanced economies.
\textit{Journal of Macroeconomics}, 50, 224--240.
\url{https://doi.org/10.1016/j.jmacro.2016.07.005}

\bibitem[Auclert et~al.(2021)]{auclert2021}
Auclert, A., Bard\'oczy, B., Rognlie, M., and Straub, L. (2021).
Using the sequence-space Jacobian to solve and estimate heterogeneous-agent models.
\textit{Econometrica}, 89(5), 2375--2408.
\url{https://doi.org/10.3982/ECTA17434}

\bibitem[Auerbach and Gorodnichenko(2012)]{auerbach2012}
Auerbach, A. J. and Gorodnichenko, Y. (2012).
Measuring the output responses to fiscal policy.
\textit{American Economic Journal: Economic Policy}, 4(2), 1--27.
\url{https://doi.org/10.1257/pol.4.2.1}

\bibitem[Blanchard and Leigh(2013)]{blanchardleigh2013}
Blanchard, O. J. and Leigh, D. (2013).
Growth forecast errors and fiscal multipliers.
\textit{American Economic Review}, 103(3), 117--120.
\url{https://doi.org/10.1257/aer.103.3.117}

\bibitem[Blanchard and Perotti(2002)]{blanchardperotti2002}
Blanchard, O. and Perotti, R. (2002).
An empirical characterization of the dynamic effects of changes in government spending and taxes on output.
\textit{The Quarterly Journal of Economics}, 117(4), 1329--1368.
\url{https://doi.org/10.1162/003355302320935043}

\bibitem[Bohn(1998)]{bohn1998}
Bohn, H. (1998).
The behavior of U.S. public debt and deficits.
\textit{The Quarterly Journal of Economics}, 113(3), 949--963.
\url{https://doi.org/10.1162/003355398555793}

\bibitem[Bom and Ligthart(2014)]{bom2014}
Bom, P. R. D. and Ligthart, J. E. (2014).
What have we learned from three decades of research on the productivity of public capital?
\textit{Journal of Economic Surveys}, 28(5), 889--916.
\url{https://doi.org/10.1111/joes.12037}

\bibitem[Calvo(1998)]{calvo1998}
Calvo, G. A. (1998).
Capital flows and capital-market crises: The simple economics of sudden stops.
\textit{Journal of Applied Economics}, 1(1), 35--54.
\url{https://doi.org/10.1080/15140326.1998.12040516}

\bibitem[Clarida et~al.(1999)]{clarida1999}
Clarida, R., Gal\'i, J., and Gertler, M. (1999).
The science of monetary policy: A New Keynesian perspective.
\textit{Journal of Economic Literature}, 37(4), 1661--1707.
\url{https://doi.org/10.1257/jel.37.4.1661}

\bibitem[Dornbusch(1976)]{dornbusch1976}
Dornbusch, R. (1976).
Expectations and exchange rate dynamics.
\textit{Journal of Political Economy}, 84(6), 1161--1176.
\url{https://doi.org/10.1086/260506}

\bibitem[Fleming(1962)]{fleming1962}
Fleming, J. M. (1962).
Domestic financial policies under fixed and under floating exchange rates.
\textit{IMF Staff Papers}, 9(3), 369--380.
\url{https://doi.org/10.2307/3866091}

\bibitem[Gechert(2015)]{gechert2015}
Gechert, S. (2015).
What fiscal policy is most effective? A meta-regression analysis.
\textit{Oxford Economic Papers}, 67(3), 553--580.
\url{https://doi.org/10.1093/oep/gpv027}

\bibitem[Gechert and Rannenberg(2018)]{gechert2018}
Gechert, S. and Rannenberg, A. (2018).
Which fiscal multipliers are regime-dependent? A meta-regression analysis.
\textit{Journal of Economic Surveys}, 32(4), 1160--1182.
\url{https://doi.org/10.1111/joes.12241}

\bibitem[Hicks(1937)]{hicks1937}
Hicks, J. R. (1937).
Mr. Keynes and the ``Classics'': A suggested interpretation.
\textit{Econometrica}, 5(2), 147--159.
\url{https://doi.org/10.2307/1907242}

\bibitem[Huidrom et~al.(2020)]{huidrom2020}
Huidrom, R., Kose, M. A., Lim, J. J., and Ohnsorge, F. L. (2020).
Why do fiscal multipliers depend on fiscal positions?
\textit{Journal of Monetary Economics}, 114, 109--125.
\url{https://doi.org/10.1016/j.jmoneco.2019.03.004}

\bibitem[Ilzetzki et~al.(2013)]{ilzetzki2013}
Ilzetzki, E., Mendoza, E. G., and V\'egh, C. A. (2013).
How big (small?) are fiscal multipliers?
\textit{Journal of Monetary Economics}, 60(2), 239--254.
\url{https://doi.org/10.1016/j.jmoneco.2012.10.011}

\bibitem[Ilzetzki et~al.(2019)]{ilzetzkireinhartrogoff2019}
Ilzetzki, E., Reinhart, C. M., and Rogoff, K. S. (2019).
Exchange arrangements entering the twenty-first century: Which anchor will hold?
\textit{The Quarterly Journal of Economics}, 134(2), 599--646.
\url{https://doi.org/10.1093/qje/qjy033}

\bibitem[International Monetary Fund(2014)]{imf2014}
International Monetary Fund. (2014).
Is it time for an infrastructure push? The macroeconomic effects of public investment.
In \textit{World Economic Outlook, October 2014: Legacies, Clouds, Uncertainties}, Chapter 3.
\url{https://www.imf.org/-/media/Websites/IMF/imported-flagship-issues/external/pubs/ft/weo/2014/02/pdf/_c3pdf.pdf}

\bibitem[Kaplan and Violante(2014)]{kaplanviolante2014}
Kaplan, G. and Violante, G. L. (2014).
A model of the consumption response to fiscal stimulus payments.
\textit{Econometrica}, 82(4), 1199--1239.
\url{https://doi.org/10.3982/ECTA10528}

\bibitem[Kaplan et~al.(2018)]{kaplan2018}
Kaplan, G., Moll, B., and Violante, G. L. (2018).
Monetary policy according to HANK.
\textit{American Economic Review}, 108(3), 697--743.
\url{https://doi.org/10.1257/aer.20160042}

\bibitem[Kraay(2012)]{kraay2012}
Kraay, A. (2012).
How large is the government spending multiplier? Evidence from World Bank lending.
\textit{The Quarterly Journal of Economics}, 127(2), 829--887.
\url{https://doi.org/10.1093/qje/qjs008}

\bibitem[Kraay(2014)]{kraay2014}
Kraay, A. (2014).
Government spending multipliers in developing countries.
\textit{American Economic Journal: Macroeconomics}, 6(4), 170--208.
\url{https://doi.org/10.1257/mac.6.4.170}

\bibitem[Leeper et~al.(2010)]{leeperwalkeryang2010}
Leeper, E. M., Walker, T. B., and Yang, S.-C. S. (2010).
Government investment and fiscal stimulus.
\textit{Journal of Monetary Economics}, 57(8), 1000--1012.
\url{https://doi.org/10.1016/j.jmoneco.2010.09.002}

\bibitem[Mountford and Uhlig(2009)]{mountforduhlig2009}
Mountford, A. and Uhlig, H. (2009).
What are the effects of fiscal policy shocks?
\textit{Journal of Applied Econometrics}, 24(6), 960--992.
\url{https://doi.org/10.1002/jae.1079}

\bibitem[Mundell(1963)]{mundell1963}
Mundell, R. A. (1963).
Capital mobility and stabilization policy under fixed and flexible exchange rates.
\textit{Canadian Journal of Economics and Political Science}, 29(4), 475--485.
\url{https://doi.org/10.2307/139336}

\bibitem[Nakamura and Steinsson(2014)]{nakamurasteinsson2014}
Nakamura, E. and Steinsson, J. (2014).
Fiscal stimulus in a monetary union: Evidence from US regions.
\textit{American Economic Review}, 104(3), 753--792.
\url{https://doi.org/10.1257/aer.104.3.753}

\bibitem[Parker et~al.(2013)]{parker2013}
Parker, J. A., Souleles, N. S., Johnson, D. S., and McClelland, R. (2013).
Consumer spending and the economic stimulus payments of 2008.
\textit{American Economic Review}, 103(6), 2530--2553.
\url{https://doi.org/10.1257/aer.103.6.2530}

\bibitem[Perotti(2005)]{perotti2005}
Perotti, R. (2005).
Estimating the effects of fiscal policy in OECD countries.
CEPR Discussion Paper No. 4842.
\url{https://cepr.org/publications/dp4842}

\bibitem[Ramey(2019)]{ramey2019}
Ramey, V. A. (2019).
Ten years after the financial crisis: What have we learned from the renaissance in fiscal research?
\textit{Journal of Economic Perspectives}, 33(2), 89--114.
\url{https://doi.org/10.1257/jep.33.2.89}

\bibitem[Ramey and Zubairy(2018)]{rameyzubairy2018}
Ramey, V. A. and Zubairy, S. (2018).
Government spending multipliers in good times and in bad: Evidence from U.S. historical data.
\textit{Journal of Political Economy}, 126(2), 850--901.
\url{https://doi.org/10.1086/696277}

\bibitem[Reinhart and Rogoff(2004)]{reinhartrogoff2004}
Reinhart, C. M. and Rogoff, K. S. (2004).
The modern history of exchange rate arrangements: A reinterpretation.
\textit{The Quarterly Journal of Economics}, 119(1), 1--48.
\url{https://doi.org/10.1162/003355304772839515}

\bibitem[Romer(2000)]{romer2000}
Romer, D. H. (2000).
Keynesian macroeconomics without the LM curve.
\textit{Journal of Economic Perspectives}, 14(2), 149--169.
\url{https://doi.org/10.1257/jep.14.2.149}

\end{thebibliography}
\end{document}